\documentclass[letterpaper, 10 pt, conference]{ieeeconf}  

\IEEEoverridecommandlockouts                              

\usepackage{amsthm}
\usepackage{graphics} 
\usepackage{epsfig} 
\usepackage{times} 
\usepackage{amsmath} 
\usepackage{amssymb}  
\usepackage[noadjust]{cite}
\usepackage{hyperref}
\usepackage{dsfont}

\newcommand{\dx}{\dot{x}}
\newcommand{\Ob}{\mathcal{O}}
\newcommand{\sgn}{\text{\normalfont sgn}}
\theoremstyle{plain}
\newtheorem{theorem}{Theorem}
\newtheorem{remark}{Remark}
\newtheorem{lemma}{Lemma}
\newtheorem{corollary}{Corollary}
\theoremstyle{definition}
\newtheorem{definition}{Definition}

\title{\LARGE \bf
Identifying Network Structure of Linear Dynamical Systems: Observability and Edge Misclassification
}

\author{Jaidev Gill and Jing Shuang (Lisa) Li
\thanks{J.G. and J.S.L. are with the Department of Electrical Engineering and Computer Science, University of Michigan, Ann Arbor, MI, 48109, USA.
        {\tt\small \{jaidevg, jslisali\}@umich.edu}.}%
}

\begin{document}

\maketitle
\thispagestyle{empty}
\pagestyle{empty}

\begin{abstract}

This work studies the limitations of uniquely identifying the structure (i.e., topology) of a networked linear system from partial measurements of its nodal dynamics.
In general, many networks can be consistent with these measurements; this is a consideration often neglected by standard network inference methods. We show that the space of these networks are related through the nullspace of the observability matrix for the true network. We establish relevant metrics to investigate this space, including an analytic characterization of the most structurally dissimilar network that can be inferred, as well as the possibility of mis-inferring presence or absence of edges. In simulations, we find that when observing over 6\% of nodes in random network models (e.g., Erd\H os-R\' enyi and Watts-Strogatz), approximately 99\% of edges are correctly classified. Extending this discussion, we construct a family of networks that keep measurements $\epsilon$-close to each other, and connect the identifiability of these networks to the spectral properties of an augmented observability Gramian.  

\end{abstract}

\section{Introduction and Motivation}
Inferring causal relationships between nodes in networks (e.g., regions in the brain) from time series measurements is broadly relevant to the engineering and scientific community \cite{STEPANIANTS_2020, BELAUSTEGUI_2024, MONTANARI_2020, YU_2007, CASTI_2023}. Various inference methods have been proposed to accomplish this task, but they typically yield varying predictions \cite{STEPANIANTS_2020}. Specifically, given some time series measurements of the nodal dynamics, different methods will infer different network structures (i.e., topologies); these are consistent with the measurements but often different from each other. Moreover, these methods typically assume that every node in the network can be measured and perturbed, an assumption that is often violated in practice for large networks where only a subset of nodes are accessible.

This question is particularly relevant for the field of neuroscience, where there is a concerted effort to understand the brain's network structure (dubbed the \textit{connectome}) based on time series measurements of brain regions and neurons/neural populations.
A key question in the field is how the network structure of the brain influences its function.
However, identifying network structure can prove challenging, as many different structures can exhibit similar behaviors \cite{LI_2025, MARDER_2011}.
This challenge can be partially addressed by systems theory literature. 
The study of parameter identifiability typically employs an input-output view \cite{GREWAL_1976, DISTEFANO_1980}; these have been particularized to linear network models \cite{WEERTS_2018}, which has led to frameworks to understand how partially observed nodes impact identifiability of these systems \cite{HENDRICKX_2019, BAZANELLA_network_2019}.
However, existing work typically makes  engineering-centric assumptions, e.g.,  assuming that we can excite all or most nodes with known perturbations.  
This is clearly infeasible in the neuroscience setting; most of the time even single perturbations are infeasible or prohibitively costly.
The focus of this paper is to study the problem of network structure identification for linear dynamical systems under passive partial measurements.\footnote{For results on network structure identification in the nonlinear setting, see companion paper \cite{GILL_2025}.}. 

\begin{figure}[t!]
      \centering
      \includegraphics[clip, trim=2.1cm 2.5cm 4.3cm 2.55cm,width = \linewidth]{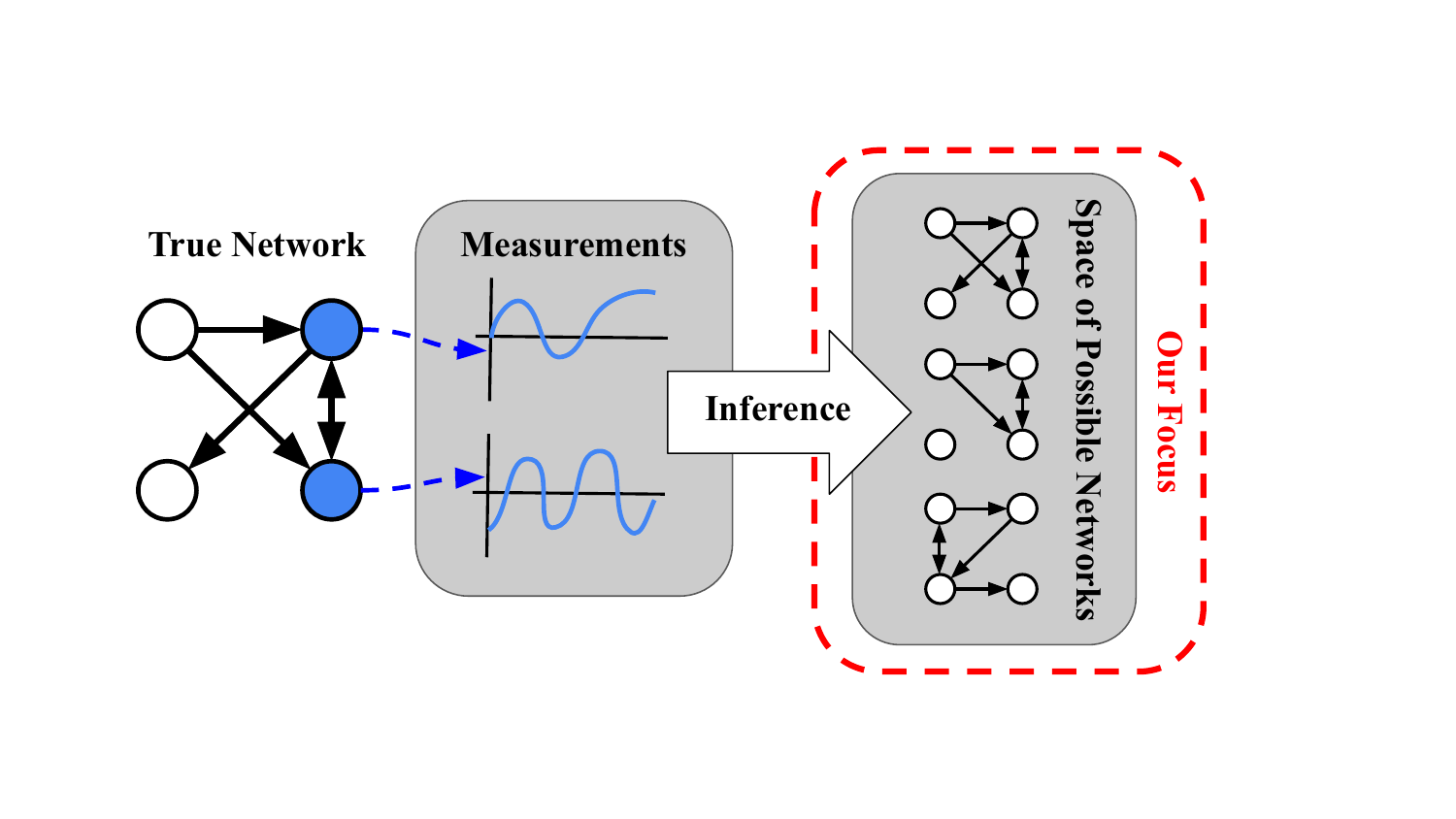}
      \caption{Given time series measurements of a subset of nodes in a network, inference methods often predict different networks and structures. This paper describes the set of possible networks consistent with these measurements and provides characterizations of this set.}
      \label{fig:overview}
   \end{figure}

In general, as implied above, many networks can be consistent with partial measurements. 
We are interested in identifying \textbf{all} such networks from given measurements (see Fig.~\ref{fig:overview}).
We start by establishing necessary and sufficient conditions under which two systems will generate identical measurements when unperturbed by external inputs (Section \ref{sec:identical}). 
We also provide a formulation to solve for the most structurally dissimilar network (i.e., how bad can our inference be) and provide conditions under which this network can be analytically computed (Section \ref{sec:min_similar}). We extend our analysis to the case where measurements are close but not necessarily identical, establishing links to the spectral properties of an augmented observability Gramian (Section \ref{sec:close_obs}).


\noindent \textbf{Notation.~} 
$A_{ij}$ denotes the element in the $i$-th row and $j$-th column of $A$.
The $\ell_{p}$ norm of a vector is denoted by $\| \cdot \|_p$. The element-wise (Hadamard) and Kronecker product by $\odot$ and $\otimes$. Let $|\cdot |$ and $\sgn(\cdot)$ be the absolute value and sign function which are applied element-wise. Let $\mathcal{N}(\cdot)$ and $\mathcal{R}(\cdot)$ correspond to the nullspace and range of a matrix. Let $\text{diag}(\cdot)$ $\big(\text{blkdiag}(\cdot)\big)$ be a diagonal (block diagonal) matrix formed from the elements of its argument.

\section{Conditions for Identical Measurements}\label{sec:identical}
 
Consider the following networked linear system:
\begin{equation}\label{sys:orig}\Sigma : \begin{cases}
    \dx(t) = Ax(t), & x(0) = x_0\\
    y(t) = Cx(t) .
\end{cases}
\end{equation}
Where $A \in \mathbb{R}^{n\times n}$ encodes the weighted adjacency matrix of the network, $x \in \mathbb{R}^n$ encodes the states of the nodes, and $y \in \mathbb{R}^p$ is a linear measurement of the states through an  observation matrix $C \in \mathbb{R}^{p\times n}$. We study the conditions under which a perturbed system,
\begin{equation}\label{sys:perturbed}
    \tilde{\Sigma} : \begin{cases}
    \dot{\tilde{x}}(t) = (A+\Delta)\tilde{x}(t), & \tilde{x}(0) = x_0\\
    \tilde{y}(t) = C\tilde{x}(t) ,
\end{cases}
\end{equation}
yields identical measurements to the original system. Here, the perturbed system evolves under $A + \Delta$, where $\Delta$ is a perturbation to the network edges. 
\begin{definition}
    The measurements of the original system \eqref{sys:orig} and the perturbed system \eqref{sys:perturbed} are \textit{indistinguishable} if $\forall x_0$, $y(t) = \tilde{y}(t)~~\forall t \geq 0$.
\end{definition}
\begin{lemma}\label{lem:equiv}
    Measurements of systems \eqref{sys:orig} and \eqref{sys:perturbed} are indistinguishable, if and only if $CA^k = C(A+\Delta)^k~~\forall k \in \mathbb{Z}_{\geq 0}$. 
\end{lemma}
\begin{proof}
    Let $y(t)$ and $\tilde{y}(t)$ denote the measurements of \eqref{sys:orig} and \eqref{sys:perturbed} respectively. Assume that the measurements are indistinguishable, i.e., $y(t) = \tilde{y}(t)~~\forall t$. Then,
    \begin{equation}
        y(t) - \tilde{y}(t)= C(e^{At} - e^{(A+\Delta)t})x_0 = 0. 
    \end{equation}
    Since $x_0$ is an arbitrary initial condition this is equivalent to
    \begin{equation}
        C(e^{At} - e^{(A+\Delta)t})= \sum_{k=0}^\infty C(A^k - (A+\Delta)^k)\frac{t^k}{k!} = 0.
    \end{equation}
     Since this equation holds $\forall t$, for necessity $(\Rightarrow)$, at time $t = 0$ the derivatives and function values must be equal:
    \begin{equation}\label{eq:equiv_cond}
    \begin{aligned}
        \frac{d^j}{dt^j}\left(\sum_{k=0}^\infty C(A^k - (A+\Delta)^k)\frac{t^k}{k!}\right)\Bigr|_{t = 0} = 0 ~~\forall j \in \mathbb{Z}_{\geq 0} &\ \\ \Leftrightarrow 
        CA^k = C(A+\Delta)^k ~~\forall k \in \mathbb{Z}_{\geq 0}. &\
        \end{aligned}
    \end{equation}
    Thus we have shown $y(t) = \tilde{y}(t) ~~\forall t \Rightarrow CA^k = C(A+\Delta)^k ~~\forall k \in \mathbb{Z}_{\geq 0}$.
    Now to establish sufficiency $(\Leftarrow )$, assume that $CA^k = C(A+\Delta)^k ~~\forall k \in \mathbb{Z}_{\geq 0}$. Then we know that
    \begin{equation} \begin{split}
        y(t) = \sum_{k=0}^\infty CA^k \frac{x_0t^k}{k!} = \sum_{k=0}^\infty C(A+\Delta)^k \frac{x_0t^k}{k!} = \tilde{y}(t),
        \end{split}
    \end{equation} 
    and thus $CA^k = C(A+\Delta)^k ~~\forall k \in \mathbb{Z}_{\geq 0} \Rightarrow y(t) = \tilde{y}(t) ~~\forall t$, 
    which concludes the proof.
\end{proof}
The following corollary shows that we only require equivalence up to the $n$-th power. 
\begin{corollary}\label{cor:C-H}
    The measurements are indistinguishable if and only if $CA^k = C(A+\Delta)^k$ for $k = 0, \cdots ,n$. 
\end{corollary}
\begin{proof}
    This corollary is a direct consequence of the Cayley-Hamilton Theorem. From Lemma \ref{lem:equiv}, $y(t) = \tilde{y}(t)~~\forall t \Leftrightarrow CA^k = C(A+\Delta)^k ~~\forall k \in \mathbb{Z}_{\geq 0}$. Notice that if $CA^k = C(A+\Delta)^k \text{ for } k = 0, \cdots, n$ we have
    \begin{multline*}
            \begin{aligned}
        CA &= C(A + \Delta)  \Rightarrow C\Delta = 0 \\
        CA^2 &= CA(A+\Delta) \Rightarrow CA\Delta = 0 \\
        &~~\vdots \\
        CA^n &= CA^{n-1}(A+\Delta) \Rightarrow CA^{n-1}\Delta = 0 .
    \end{aligned}
    \end{multline*}
    Hence, $CA^k\Delta = 0$ for $k = 0, \cdots, n-1$. For all $k \geq n$, we can apply the Cayley-Hamilton theorem to express $CA^k\Delta$ as a linear combination of its lower powers, i.e., $CA^k\Delta = \sum_{m=0}^{n-1}\alpha_mCA^m\Delta$ for some $\alpha_m \in \mathbb{R}$. By the above analysis, all terms in this sum are equal to zero. Thus, we may conclude that $CA^k\Delta = 0 \quad \forall k\geq 0$.  
\end{proof}

The result of Lemma \ref{lem:equiv} and Corollary \ref{cor:C-H} allow us to relate indistinguishability of \eqref{sys:orig} and \eqref{sys:perturbed} to observability of \eqref{sys:orig}. 

\begin{theorem}\label{thm:identical}
    Let $\Ob$ be the observability matrix of the original system, $\Sigma$ given in \eqref{sys:orig}.
    Measurements will be indistinguishable if and only if $\Ob \Delta = 0$.  
\end{theorem}

\begin{proof}
     From Corollary \ref{cor:C-H} we know that, 
     $ y(t) = \tilde{y}(t)~~\forall t \Leftrightarrow
        CA^k = C(A+\Delta)^k ~~ \text{ for } k = 0, \cdots, n$.
      For $ k = 1 $ we see that this condition is satisfied if $C\Delta = 0$; for $k = 2$ we require $CA\Delta = 0$. Repeated evaluation establishes
    \begin{equation}
        \begin{bmatrix}
            C^\top &
            (CA)^\top &
            \dots &
            (CA^{n-1})^\top
        \end{bmatrix}^\top \Delta = \Ob \Delta = 0, 
    \end{equation}
    which concludes the proof. 
\end{proof}

\begin{remark}\label{rem:col-wise}
     Theorem \ref{thm:identical} is a column-wise constraint on the perturbations $\Delta$. This implies that modifications of the outgoing edges of a node are independent from other nodes. 
\end{remark}

Theorem \ref{thm:identical} implies the set of networks that are indistinguishable is $\mathcal{A} = \{A+\Delta : \Ob \Delta = 0  \}$. Consider the following example network and measurement scheme:
\begin{equation} \label{eq:example}
    A = \begin{bmatrix}
        1 & 1 & 1 & 0 \\
        0 & 1 & 1 & 0 \\
        1 & 0 & 0 & 0 \\
        0 & 1 & 1 & 1
    \end{bmatrix}, ~~ C = \begin{bmatrix}
        1 & 0 & 0 & 0 
    \end{bmatrix}.
\end{equation}
The number of \textit{structurally dissimilar} networks, i.e., networks that do not share the same underlying unweighted adjacency matrix, is $2^{n^2}$. For a network with $n=4$ nodes, as in \eqref{eq:example}, there are $65536$ possible structurally dissimilar networks. For the choice of $A$ and $C$ 
in \eqref{eq:example}
\begin{equation}
    \mathcal{N}(\Ob)=  \text{span}\bigl\{\begin{bmatrix}
        0 & -1 & 1 & 0 
    \end{bmatrix}^\top, \begin{bmatrix}
        0& 0 & 0 & 1
    \end{bmatrix}^\top \bigr\}.
\end{equation}
Consider the first column of $A$, which can be modified by adding a linear combination of vectors in $\mathcal{N}(\Ob)$ into:
\begin{equation}\label{eq:possible_col}
    \begin{bmatrix}
        1 \\ 1 \\ 0 \\ 0
    \end{bmatrix}, \begin{bmatrix}
        1 \\ 1 \\ 0 \\ 1
    \end{bmatrix},
    \begin{bmatrix}
        1 \\ 2 \\ -1 \\ 0
    \end{bmatrix},
    \begin{bmatrix}
        1 \\ 2 \\ -1 \\ 1
    \end{bmatrix},
    \begin{bmatrix}
        1 \\ 0 \\ 1 \\ 1
    \end{bmatrix}.
\end{equation}
Other values are possible; we have shown the five that represent all possible structures (i.e., sparsity patterns). Notice that edge $A_{11} = 1$ is present in all structures in $\mathcal{A}$. In fact, all incoming edges into node 1 ($A_{1j} ~ \forall j$) cannot be altered, motivating a formalization of this phenomenon:
\begin{definition}
    Let the columns of $\Phi$ form a basis for $\mathcal{N}(\Ob)$ and $e_i$ be a standard basis vector with $1$ in the $i$-th component. 
    An edge encoded in $A_{ij}$ is \textit{structurally essential} if 
    \begin{equation}
        \Phi^\top e_i = 0.
    \end{equation}
    
\end{definition}
\noindent In our example, all edges in $A_{1j}$ are structurally essential and cannot be altered without changing measured dynamics. 

\begin{definition}
    An edge encoded in $A_{ij}$ is \textit{structurally decoupled} from all other edges if $\exists v$ such that
    \begin{equation}
        \Phi v = e_i.
    \end{equation}
    We refer to edges that are not structurally decoupled as being \textit{structurally coupled} to some other edges. 
\end{definition}

\noindent Notice that edge $A_{41}$ is structurally decoupled from all other edges; thus, we can freely remove or add this edge. 
Here, enumerating all possible network structures amounts to determining whether the basis of $\mathcal{N}(\Ob)$ allows the structurally coupled edges in the network to be removed independently of one another. 
Some structurally coupled edges may only be removed or added simultaneously.
For example, edges $A_{24}$ and $A_{34}$ must both be present or both missing; similarly, \eqref{eq:possible_col} shows that there is no way to simultaneously remove both $A_{21}$ and $A_{31}$. In this particular example, there are $6^3 \times 4 = 864$ structurally dissimilar networks in $\mathcal{A}$, emphasizing the unwieldy nature of enumerating this set. The combinatorial iterations through the various structures motivates developing aggregate characterizations of $\mathcal{A}$. 

\section{Dissimilar Networks}\label{sec:min_similar}

Given some true network $A$, we would like to understand the structure of the most dissimilar network that is consistent with the measurements made. Towards understanding this ``worst-case" network we formulate a minimization program over the set $\mathcal{A}$. Defining  $Z \in \mathbb{R}^{n\times n}$ such that $Z_{ij} = 1$ if $A_{ij} \neq 0$ and $Z_{ij} = 0$ otherwise, consider:
\begin{equation} \label{eq:min_l1}
    \begin{aligned}
        \min_{\Delta} &\ \|\text{vec}( Z \odot (A+\Delta))\|_{1}\\
        \textrm{s.t.} &\ \Ob \Delta = 0.
    \end{aligned}
\end{equation}
The element-wise product of $Z$ and $A+\Delta$ alongside the use of an $\ell_1$ norm encourages the solution of \eqref{eq:min_l1} to remove edges that exist in $A$.

Now, let  $\text{rank}(\Ob) = r$ hence, $\dim(\mathcal{N}(\Ob)) = n-r$ and thus, we can define $\Phi \in \mathbb{R}^{n\times (n-r)}$ such that the columns of $\Phi$ form an orthonormal basis for $\mathcal{N}(\Ob)$ such that $\Phi^\top \Phi = I_r$. Parameterizing $\Delta = \Phi V$ where $V \in \mathbb{R}^{(n-r)\times n}$, we can write the program in \eqref{eq:min_l1} as an unconstrained program
\begin{equation} \label{eq:min_l1_unconstrained}
    \begin{aligned}
        \min_{V} &\ \| \text{vec}(A+Z\odot(\Phi V))\|_{1}.\\
    \end{aligned}
\end{equation}
Identifying that \begin{equation}
    \text{vec}(A+Z\odot(\Phi V)) = a + \text{diag}(z)(I_n\otimes \Phi)v
\end{equation}
where $a = \text{vec}(A) \in \mathbb{R}^{n^2}$, $v = \text{vec}(V)\in \mathbb{R}^{n(n-r)}$ and $z = \text{vec}(Z)\in \mathbb{R}^{n^2}$, we can rewrite \eqref{eq:min_l1} using the epigraph form of the problem to arrive at a linear program (LP) \cite{BOYD__2004}:

\begin{equation} \label{eq:LP}
    \begin{aligned}
        \min_{t, v} &\ \mathds{1}^\top t\\
        \textrm{s.t.}&\ -t \leq a + \Gamma v \leq t.
    \end{aligned}
\end{equation}
Here $\Gamma = \text{diag}(z)(I_n\otimes \Phi) \in \mathbb{R}^{n^2\times n(n-r)}$.


Problem \eqref{eq:LP} can be readily solved with LP solvers. We now provide conditions under which this problem admits a closed form analytic solution; this provides better intuition about the problem and is also advantageous for the analysis of very large networks (which translate to very large LPs). 

\subsection{Equivalence of the $\ell_{1}$ and $\ell_{2}$ Programs}

Consider the program given in \eqref{eq:min_l1_unconstrained}; we now relax it to instead minimize the $\ell_2$ norm:
\begin{equation} \label{eq:min_l2}
    \begin{aligned}
        \min_{V} &\ \frac{1}{2}\| \text{vec}(A+Z\odot(\Phi V))\|_2^2.
    \end{aligned}
\end{equation}
This problem is convex; to find the optimal solution, take the gradient of the objective  and set it to zero:

\begin{equation}\label{eq:opt_l2}
    \Phi^\top(A + Z\odot(\Phi V^*)) = 0.
\end{equation}

Note that we can vectorize this to yield a linear equation, i.e., $(I_n\otimes\Phi^\top)(a + \Gamma v^*) = 0$. Given the solution $V^*$ to matrix equation \eqref{eq:opt_l2}, we now provide a condition under which it is also a solution to \eqref{eq:min_l1}.  

\begin{theorem}\label{thm:equiv_1_2}
Let $V^*$ be a minimizer for \eqref{eq:min_l2}. It is also a minimizer to the original program \eqref{eq:min_l1} if \begin{equation} \label{eq:Thm2_minimizer_condition}
    \Phi^\top(Z\odot \sgn(A + Z\odot (\Phi V^*))) = 0.
\end{equation} 
\end{theorem}

\begin{proof}
    Let $v^* = \text{vec}(V^*)$ where $V^*$ is a minimizer of \eqref{eq:min_l2}. Assume that $V^*$ satsifies \eqref{eq:Thm2_minimizer_condition}.
    Recall that the programs \eqref{eq:min_l1} and \eqref{eq:LP} are equivalent. We now show that for an appropriate choice $(v^*, t^*, \lambda^*)$, where $\lambda^* =  \begin{bmatrix}
        {\lambda_1^*}^\top & {\lambda_2^*}^\top
    \end{bmatrix}^\top$ are the dual variables associated with the constraints in \eqref{eq:LP}, the Karush-Kuhn-Tucker (KKT) conditions for \eqref{eq:LP} are satisfied. 
    We select $t^* = |a + \Gamma v^*|$, and $\lambda_{1,i}^* = 1$ and $\lambda_{2,i}^* = 0$ if $[\sgn(a + \Gamma v^*)]_i = 1$, $\lambda_{1,i}^* = 0$ and $\lambda_{2,i}^* = 1$ if $[\sgn(a + \Gamma v^*)]_i = -1$, and lastly $\lambda_{1,i}^* = 1/2 = \lambda_{2,i}^* $ if $[\sgn(a + \Gamma v^*)]_i = 0$. With these choices $\lambda_1^* - \lambda_2^* = \sgn(a + \Gamma v^*)$. 
    The Lagrangian of \eqref{eq:LP} is
    \begin{multline}
        \mathcal{L}(v,t,\lambda) = \mathds{1}^\top t +  \lambda^\top \left(\begin{bmatrix}
            \Gamma & -I_{n^2} \\
            -\Gamma & -I_{n^2}
        \end{bmatrix}\begin{bmatrix}
            v \\ t
        \end{bmatrix} + \begin{bmatrix}
            a \\ -a
        \end{bmatrix}\right).
    \end{multline}
The corresponding KKT conditions are

      \begin{subequations}
        \begin{align}
            \begin{bmatrix}
                0 \\ \mathds{1}
            \end{bmatrix} + \begin{bmatrix}
            \Gamma ^\top& -\Gamma^\top \\
            -I_{n^2} & -I_{n^2}
        \end{bmatrix}\lambda^* = 0 \label{eq:stationarity}\\
        \begin{bmatrix}
            \Gamma & -I_{n^2} \\
            -\Gamma & -I_{n^2}
        \end{bmatrix}\begin{bmatrix}
            v^* \\ t^*
        \end{bmatrix} + \begin{bmatrix}
            a \\ -a
        \end{bmatrix} \leq 0 \label{eq:feasibility}\\
        \lambda^* \geq 0 \label{eq:geq}\\
        \text{diag}(\lambda^*)\left(\begin{bmatrix}
            \Gamma & -I_{n^2} \\
            -\Gamma & -I_{n^2}
        \end{bmatrix}\begin{bmatrix}
            v^* \\ t^*
        \end{bmatrix} + \begin{bmatrix}
            a \\ -a
        \end{bmatrix}\right) = 0. \label{eq:comp_slack}
        \end{align}
    \end{subequations}
    With $t^*$ as chosen, the feasibility condition \eqref{eq:feasibility} is satisfied. Moreover, \eqref{eq:geq} is satisfied by our choice of $\lambda^*$. It remains to check \eqref{eq:stationarity} and \eqref{eq:comp_slack}. Notice, from the stationarity of the Lagrangian \eqref{eq:stationarity} we have that $\lambda_1^* + \lambda_2^* = \mathds{1}$ which is also satisfied from our choice of $\lambda^*$. 
    Rewriting \eqref{eq:comp_slack} as\begin{equation}\label{eq:broken_slack}\begin{split}
        \text{diag}(\lambda_1^*)(a + \Gamma v^* - |a + \Gamma v^*|) = 0  \\
        \text{diag}(\lambda_2^*)(-(a + \Gamma v^*) - |a + \Gamma v^*|)= 0, \end{split}
    \end{equation}
    adding and subtracting these equations to each other we have 
    \begin{equation}\begin{split}
    \text{diag}(\lambda_1^* - \lambda_2^*)(a + \Gamma v^*) = |a + \Gamma v^*| \\
    \text{diag}(\lambda_1^* - \lambda_2^*)|a + \Gamma v^*| = a + \Gamma v^*. \end{split}
\end{equation}
Since $\lambda_1^* - \lambda_2^* = \sgn(a + \Gamma v^*)$ these equations are satisfied, and thus \eqref{eq:comp_slack} is satisfied.
Lastly, we must consider the first $n(n-r)$ rows of the stationarity of the Lagrangian \eqref{eq:stationarity}, 
\begin{equation}
    \Gamma ^\top(\lambda_1^*-\lambda_2^*) = \Gamma ^\top\sgn(a + \Gamma v^*)  = 0.
\end{equation}
Which in matrix form reads as
\begin{equation}
    \Phi^\top(Z\odot \sgn(A + Z\odot (\Phi V^*)))= 0,
\end{equation}
which is satisfied by assumption and concludes the proof. 

\end{proof}

\subsection{Illustrative Example}

We now revisit example \eqref{eq:example}. In the left half of Fig.~\ref{fig:identical_examples}, we display the original network (i.e., \eqref{eq:example}); in the right half of the figure, we display the minimizer to \eqref{eq:min_l1}, which is the most structurally dissimilar network that still exhibits identical measurements as the original network.
In this example, the condition in Theorem \ref{thm:equiv_1_2} is satisfied, so it suffices to solve \eqref{eq:opt_l2} to find the most dissimilar network.

Not counting self-dynamics, the two networks have only two edges in common. Since all edges ($A_{4j}$) into node 4 are structurally decoupled, they are entirely removed in the dissimilar network.
Notice that since nodes 2 and 3 influence the dynamics of node 1, the dissimilar network  essentially permutes the role of nodes 2 and 3 compared to the original network.
Conversely, structurally essential edges $A_{1j}$ remain present in both networks, as expected. 
 
This example illustrate a few key points.
First, edges connected to nodes that do not influence measured nodes can be removed from the network. Second, the relationship between unmeasured but influential nodes can be changed by a simple permutation of their labels. 


\begin{figure}[thpb]      
      \centering
      \begin{minipage}{0.6\linewidth}
      \centering
          \includegraphics[clip, trim=1.27cm 4.5cm 14.3cm 4.5cm, scale = 0.45]{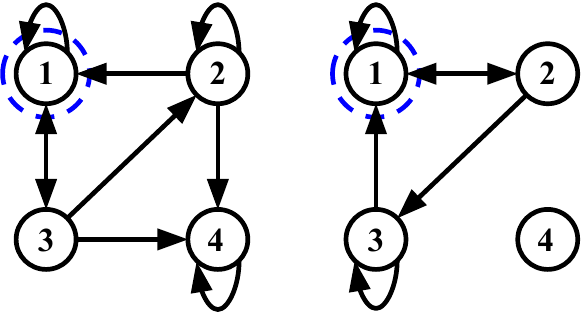}
      \end{minipage}\hfill
      \begin{minipage}{0.4\linewidth}
            \caption{(Left) Network described by \eqref{eq:example}. (Right) The most structurally dissimilar network exhibiting identical measurements as the original network. Blue dashed circles indicate measured nodes.}\label{fig:identical_examples}
      \end{minipage}
   \end{figure}


\section{Networks with $\epsilon$-Close Measurements}\label{sec:close_obs}
In practice, measurements are corrupted by noise, and executing the same experiment twice will yield similar yet non-identical measurements. 
We now extend previous analysis to account for this: we characterize the set of networks who yield similar (not necessarily identical) measurements to the original system.
We assume both $A$ and $A+\Delta$ are Hurwitz\footnote{In practice this assumption is necessary as unstable systems would yield exponentially growing trajectories; it is less meaningful to compare the ``closeness" of such trajectories.}
so that the corresponding observability Gramians are well-defined. 
We consider an augmented system akin to \cite{PAPACHRISTODOULOU_2010, ANDERSON_2011}: 
\begin{equation}
    \bar{A} = \text{blkdiag}(A, A+\Delta), ~~ \bar{C} = \begin{bmatrix}
        C & -C 
    \end{bmatrix},
\end{equation} and the system $\bar\Sigma$ as
\begin{equation}\label{sys:aug}\bar\Sigma : \begin{cases}
    \dot{\bar x}(t) = \bar{A}\bar x(t), & \bar x(0) = \begin{bmatrix}
        x_0^\top & x_0^\top
    \end{bmatrix}^\top\\
    e(t) = \bar{C}\bar x(t). 
\end{cases}
\end{equation}
Here, $\|e(t)\|_2$ is a measure of the cumulative error between measurements of the original and perturbed systems.

\begin{theorem}\label{thm:close_obsv}
    Let $\epsilon > 0$ and $\bar{W}_o$ be the observability Gramian for the augmented system $\bar{\Sigma}$ given in \eqref{sys:aug}. 
    
    If $\bar{W}_o \prec \frac{\epsilon^2}{2\| x_0\|^2_2}I$,  then $\| e(t)\|_2 < \epsilon$.
\end{theorem}
\begin{proof}
    Assume $\bar{W}_o \prec \frac{\epsilon^2}{2\| x_0\|^2_2}I$. Then,
    \begin{equation} \begin{split}
    \| e(t) \|_2^2 & = \int_0^\infty e(t)^\top e(t)dt  \\ & = \int_0^\infty\bar{x}(0)^\top  e^{\bar A ^\top t} \bar{C}^\top \bar{C} e^{\bar A  t}\bar{x}(0)dt = \bar{x}(0)^\top \bar{W}_o\bar{x}(0).
    \end{split}
\end{equation}
By assumption, $\bar{x}(0)^\top \bar{W}_o\bar{x}(0) < \bar{x}(0)^\top \frac{\epsilon^2}{2\| x_0\|^2_2}I\bar{x}(0) = \epsilon^2$. Hence,
    $\| e(t) \|_2 < \epsilon.$
\end{proof}

Theorem \ref{thm:close_obsv} elucidates the relationship between $\bar{W}_o$, the observability Gramian of \eqref{sys:aug}, and how close the measurements of the two systems will be. 
Now, we leverage it to consider the most structurally dissimilar network that produces $\epsilon$-close measurements as the original network.
This is similar to Section \ref{sec:min_similar}; however, we no longer require identical measurements. We first establish the following definition. 

\begin{definition}
    We say that $\bar{W}_o$ is a \textit{feasible} observability Gramian if $\exists \Delta$ such that
    \begin{equation}\label{eq:lyap}
    \bar{W}_o\bar{A} + \bar{A}^\top \bar{W}_o = - \bar{C}^\top \bar{C}.
    \end{equation}
\end{definition}

The following program generates the maximally dissimilar network (analogous to \eqref{eq:min_l1}) for $\epsilon$-close measurements:
\begin{equation}\label{eq:min_similar}
    \begin{aligned}
        \min_{\Delta, \bar{W}_o} &\ \|\text{vec}( Z \odot (A+\Delta))\|_{1}\\
        \textrm{s.t.} &\ \eqref{eq:lyap}, \quad
        \ \bar{W}_o \prec \frac{\epsilon^2}{2 \| x_0\|_2^2}I  .
    \end{aligned}
\end{equation}

\begin{remark}
    This program is nonconvex due to its bilinear constraint, but can be solved with branch and bound algorithms in YALMIP \cite{LOFBERG_2004}.
\end{remark}
 

We now consider the implications of \eqref{eq:lyap}.
This equation characterizes all possible Gramians $\bar{W}_o$ that correspond to networks with $\epsilon$-close measurements to the original.
However, given a feasible $\bar{W}_o$ there can exist multiple values of $\bar{A}$, which translate to multiple candidate networks via the perturbation $\Delta$.
We extend techniques from \cite{FERNANDO_1981} to determine what perturbations $\Delta$ are consistent with a given $\bar{W}_o$.

Since $\bar{W}_o$ is a real symmetric matrix, we can diagonalize it using an orthonormal matrix $V$: $\bar{W}_o = VDV^\top$. We can partition $V$ as $V = \begin{bmatrix}
    V_o & V_{\bar{o}}
\end{bmatrix}$, where $l := \text{dim}(\mathcal{N}(\bar{W}_o))$, $V_o \in \mathbb{R}^{2n \times (2n-l)}$, and $V_{\bar{o}} \in \mathbb{R}^{2n \times l}$.
The columns of $V_o$ and $V_{\bar{o}}$ form a basis for the observable and the unobservable subspace of $\bar\Sigma$, respectively. 
Diagonal matrix $D$ can be partitioned as $D = \text{blkdiag}(\Lambda, 0)$ where $\Lambda \in \mathbb{R}^{(2n-l) \times (2n-l)}$ is full rank. We now pre- and post-multiply \eqref{eq:lyap} by $V^\top$ and $V$:
\begin{equation}\label{eq:expansion}
    DV^\top\bar{A}V + V^\top\bar{A}^\top VD= - (\bar{C}V)^\top (\bar{C}V).
\end{equation}
Here, \begin{equation}
    V^\top\bar{A} V = \begin{bmatrix}
        \check{A}_o & 0 \\ \check{A}_{21} & \check{A}_{\bar{o}} 
    \end{bmatrix},~~  \bar{C}V = \begin{bmatrix}
        \check{C}_o & 0
    \end{bmatrix}
\end{equation} correspond to an observability decomposition of $\bar{\Sigma}$; so \eqref{eq:expansion} reduces to


\begin{equation}
    \Lambda\check{A}_o + \check{A}_o^\top \Lambda= -\check{C}_o^\top \check{C}_o.
\end{equation}
Following \cite{FERNANDO_1981}, pre- and post-multiply by $\Lambda^{-1/2}$ to obtain
\begin{equation}
    \Lambda^{1/2}\check{A}_o\Lambda^{-1/2} + \Lambda^{-1/2}\check{A}_o^\top \Lambda^{1/2}= -\Lambda^{-1/2}\check{C}_o^\top \check{C}_o\Lambda^{-1/2}.
\end{equation}
This implies that the symmetric part of $\Lambda^{1/2}\check{A}_o\Lambda^{-1/2}$ is $-\frac{1}{2}\Lambda^{-1/2}\check{C}_o^\top \check{C}_o\Lambda^{-1/2}$. In general, we have that
\begin{equation}
    \Lambda^{1/2}\check{A}_o\Lambda^{-1/2} = -\frac{1}{2}\Lambda^{-1/2}\check{C}_o^\top \check{C}_o\Lambda^{-1/2} + S,
\end{equation}
for an arbitrary skew-symmetric matrix $S$. Pre- and post-multiply again to yield
\begin{equation}
    \check{A}_o = -\frac{1}{2}\Lambda^{-1}\check{C}_o^\top \check{C}_o + \Lambda^{-1/2}S\Lambda^{1/2}.
\end{equation}
Returning to the full augmented system, we have
\begin{equation}
\bar{A} = 
    V \left(
  \begin{bmatrix} -\frac{1}{2}\Lambda^{-1}\check{C}_o^\top \check{C}_o+ 
    \Lambda^{-1/2}S\Lambda^{1/2} & 0 \\
        \check{A}_{21} &  \check{A}_{\bar{o}} 
    \end{bmatrix} \right)V^\top,
\end{equation}
which can be equivalently expressed as 
\begin{multline}\label{eq:set_similar}
\bar{A} = \text{blkdiag}(A, A+\Delta) = 
\underbrace{
       -\frac{1}{2}V_{o}\Lambda^{-1}\check{C}_o^\top \check{C}_oV_{o}^\top}_{\text{fixed among all networks}} + \\V_{o} \Lambda^{-1/2}S\Lambda^{1/2}V_o^\top + V_{\bar{o}} \check{A}_{12}V_{o}^\top + V_{\bar{o}}\check{A}_{\bar{o}}V_{\bar{o}}^\top.
\end{multline}
From this analysis, we see that solutions $(S, \check{A}_{21}, \check{A}_{\bar{o}})$ to \eqref{eq:set_similar} will characterize perturbations $\Delta$ and resulting networks $A+\Delta$ that generate similar measurements. 


We also see from \eqref{eq:set_similar} that for a given feasible observability Gramian, all networks share a component related to the observable subsystem of $\bar\Sigma$ given in \eqref{sys:aug}, and differ by the skew-symmetric matrix $S$ and the two free matrices $\check{A}_{21}$ and $\check{A}_{\bar o}$. To illustrate some of the quantitative trends, we consider two cases. If $\Lambda_{ii} \ll 1$, the fixed term will be dominant, thus
\begin{equation}
    \bar{A} \approx -\frac{1}{2}V_{o}\Lambda^{-1}\check{C}_o^\top \check{C}_oV_{o}^\top.
\end{equation}
This implies that only the observable term will contribute to the set of possible perturbed networks, which is consistent with our previous intuition from Section \ref{sec:identical} that the observable portions of the networks must be consistent.
Similarly we see that if we tolerate significant deviations in measurements, $\Lambda_{ii} \gg 1$ we have
\begin{equation}
    \bar{A}\approx
    V_{o} \Lambda^{-1/2}S\Lambda^{1/2}V_o^\top + V_{\bar{o}} \check{A}_{12}V_{o}^\top + V_{\bar{o}}\check{A}_{\bar{o}}V_{\bar{o}}^\top,
\end{equation}
which implies that there can be significant variability in the network structure due to the free variables' influence. 

\section{Simulations}\label{sec:simulations}

\subsection{Identical Measurements with Random Networks}

To explore the impact of partial measurement on network structure identification, we study edge misclassification in Erd\H os-R\' enyi and Watts-Strogatz random networks; the latter exhibits properties similar to brain networks \cite{BASSETT_2017}. We consider networks with $n=100$ nodes with edge presence probability $p = 1/6$ and edge flip probability $\beta = 1/30$. For the Watts-Strogatz networks, we connect each node to $K=3$ of its neighbors. We sample from a standard normal distribution to select edge weights, and randomly select which nodes to measure. 

We compute the maximally dissimilar networks (that exhibit identical dynamics) by solving \eqref{eq:LP}. In each case, we compute the percentage of edges that have been flipped (removed/added) relative to the original network; we deemed an edge present if the weight was greater $10^{-5}$. Results were averaged over $100$ trials with different random seeds and are depicted in Fig.~\ref{fig:random_networks}. When measuring under $6$ nodes, the most dissimilar networks effectively change the entire network structure without affecting the measurements. However, after observing greater than $6$ nodes, there is a significant decline in the percentage of edges flipped for both network models. This indicates a phase transition from completely ``unobservable'' networks to completely ``observable" networks.

\begin{figure}[thpb]
      \centering
      \includegraphics[clip, trim=0cm 0cm 0cm 0cm,width = 1\linewidth]{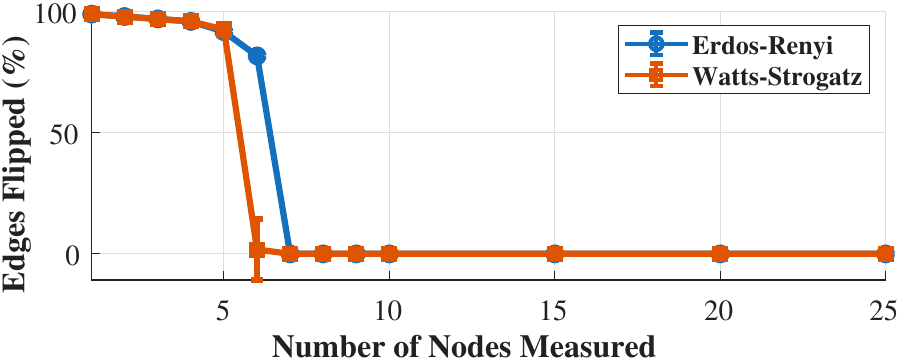}
      \caption{Percentage of edges flipped as more nodes are measured for Erd\H os-R\' enyi and Watts-Strogatz random networks.}
      \label{fig:random_networks}
   \end{figure}


\subsection{$\epsilon$-Close Measurements}
To explore the implications of the program in \eqref{eq:min_similar}, we consider the following network and observation scheme
\begin{equation}
A = \begin{bmatrix}
    -3 & 1 & 0 \\
    0 & -3 & 0 \\
    1 & 0 & -3
\end{bmatrix}, ~~ C = \begin{bmatrix}
    1 & 0 & 0 
\end{bmatrix}.   
\end{equation}
For example's sake, we seek an exact solution; thus, instead of
solving nonconvex problem \eqref{eq:min_similar}, we work backwards. We fix $\bar{W}_o$ to ensure feasibility and temporarily remove the other constraint. Then, we find that $\sqrt{\lambda_{\max}(\bar{W}_o)} = 0.65$ and the corresponding maximally structurally dissimilar network is
\begin{equation}
    A + \Delta = \begin{bmatrix}
    -2 & 0 & 0.1 \\
    0 & 0 & 0 \\
    0.833 & 0 & -2
\end{bmatrix}.
\end{equation}
In Fig.~\ref{fig:similar_obsv}, we show the observed trajectories $y(t)$ and $\tilde{y}(t)$ for the original system and most dissimilar system, respectively. We selected the initial condition so that $\| x_0\|_2 = 1$. We see that the trajectories remain close over the entire simulation; the error norm is $\| e(t)\|_2 = 0.036$. This is less than $0.65\sqrt{2}$, which is the upper bound provided by Theorem 3. In Fig.~\ref{fig:similar_nets}, we display the corresponding network structures. We see that the general intuition developed in Section \ref{sec:identical} is violated as node 2, despite influencing measured node 1, has all connections removed; this is compensated by node 3's influence on node 1. For comparison, we also include the most dissimilar network that produces identical measurements; we see that allowing the measurements to be close instead of identical yields significant structural variations.

\begin{figure}[thpb]
      \centering
      \includegraphics[clip, trim=0cm 0cm 0cm 0cm,width = 1\linewidth]{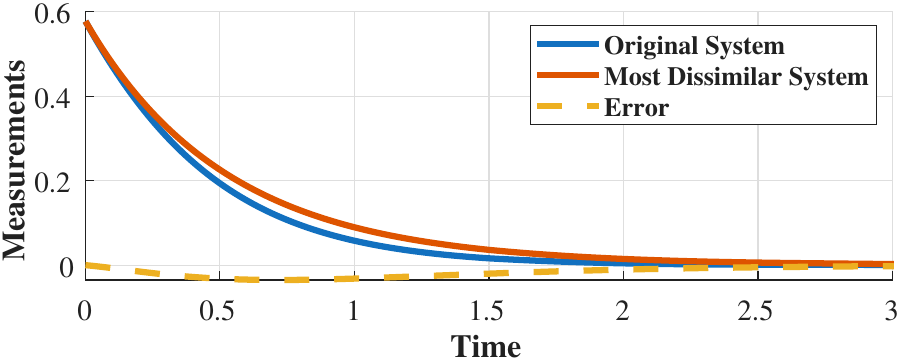}
      \caption{Comparison of the observed trajectories $y(t)$ and $\tilde{y}(t)$ and corresponding error $e(t)$ of the $\epsilon$-close network.}
      \label{fig:similar_obsv}
   \vspace{-1em}
   \end{figure}
\begin{figure}[thpb]      
      \centering
          \includegraphics[clip, trim=1.2cm 4.5cm 8.3cm 4cm, width = 0.8\linewidth]{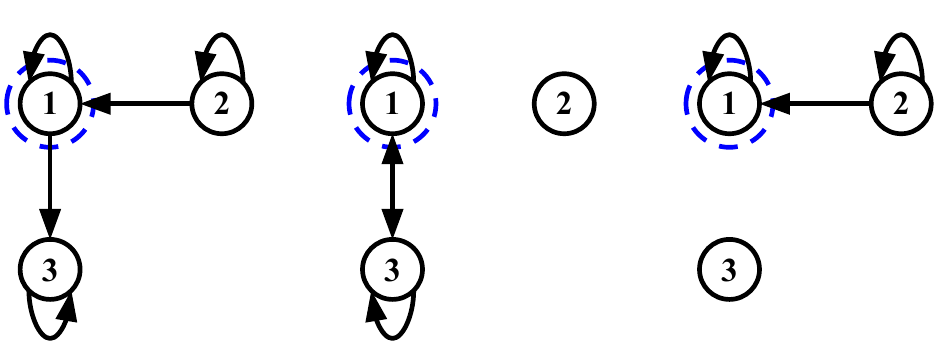}
            \caption{(Left) Original network. (Middle) The most structurally dissimilar network that generates similar measurements. (Right) The most structurally dissimilar network that generates identical measurements. Blue dashed circles indicate which nodes are measured.}\label{fig:similar_nets}
   \end{figure}

\section{Conclusions}\label{sec:conclusion}

In this work, we provided tools to understand the structural discrepancies that may arise during network inference from passive partial measurements.
These results provide a direct link between the set of networks that are consistent with observed measurements and the observability properties of the network in question. 
Future work will utilize the presented results to design new inference algorithms that produce the set of all possible data-consistent networks as opposed to a single possible network.




\bibliographystyle{IEEEtran}

\bibliography{IEEEabrv,IEEEexample}

@STRING{IEEE_J_AC         = "{IEEE} Trans. Automat. Contr."}

@article{ANDERSON_2011,
	title        = {Model decomposition and reduction tools for large-scale networks in systems biology},
	author       = {James Anderson and Yo-Cheng Chang and Antonis Papachristodoulou},
	year         = 2011,
	journal      = {Automatica},
	volume       = 47,
	number       = 6,
	pages        = {1165--1174},
	doi          = {https://doi.org/10.1016/j.automatica.2011.03.010},
	issn         = {0005-1098},
	note         = {Special Issue on Systems Biology}
}

@inproceedings{PAPACHRISTODOULOU_2010,
	title        = {Structured model reduction for dynamical networked systems},
	author       = {Papachristodoulou, Antonis and Chang, Yo-Cheng and August, Elias and Anderson, James},
	year         = 2010,
	booktitle    = {49th IEEE Conference on Decision and Control (CDC)},
	pages        = {2670--2675},
	doi          = {10.1109/CDC.2010.5718017}
}

@inproceedings{BELAUSTEGUI_2024,
	title        = {Sparse Dynamic Network Reconstruction Through L1-regularization of a Lyapunov Equation},
	author       = {Belaustegui, Ian Xul and Arango, Marcela Ordorica and Rossi-Pool, Roman and Leonard, Naomi Ehrich and Franci, Alessio},
	year         = 2024,
	booktitle    = {2024 IEEE 63rd Conference on Decision and Control (CDC)},
	volume       = {},
	number       = {},
	pages        = {4544--4549},
	doi          = {10.1109/CDC56724.2024.10885791}
}

@article{FERNANDO_1981,
	title        = {Solution of Lyapunov Equation for the State Matrix},
	author       = {K.V. Fernando and H. Nicholson},
	year         = 1981,
	journal      = {Electronics Letters},
	volume       = 17,
	number       = 5,
	pages        = {204--205}
}

@article{STEPANIANTS_2020,
	title = {Inferring causal networks of dynamical systems through transient dynamics and perturbation},
	volume = {102},
	doi = {10.1103/PhysRevE.102.042309},
	number = {4},
	journal = {Physical Review E},
	author = {Stepaniants, George and Brunton, Bingni W. and Kutz, J. Nathan},
	month = oct,
	year = {2020},
	note = {Publisher: American Physical Society},
	pages = {042309}
}

@book{BOYD__2004,
	title        = {Convex Optimization},
	author       = {Boyd, Stephen and Vandenberghe, Lieven},
	year         = 2004,
	publisher    = {Cambridge University Press},
	place        = {Cambridge}
}

@article{MONTANARI_2020,
	title        = {Observability of {Network} {Systems}: {A} {Critical} {Review} of {Recent} {Results}},
	shorttitle   = {Observability of {Network} {Systems}},
	author       = {Montanari, Arthur N. and Aguirre, Luis A.},
	year         = 2020,
	month        = dec,
	journal      = {Journal of Control, Automation and Electrical Systems},
	volume       = 31,
	number       = 6,
	pages        = {1348--1374},
	doi          = {10.1007/s40313-020-00633-5},
	issn         = {2195-3899},
	language     = {en}
}

@INPROCEEDINGS{LI_2025,
  author={Li, Jing Shuang},
  booktitle={2025 American Control Conference (ACC)}, 
  title={Toward Neuronal Implementations of Delayed Optimal Control}, 
  year={2025},
  volume={},
  number={},
  pages={2715-2721},
  doi={10.23919/ACC63710.2025.11107898}}

@inproceedings{CASTI_2023,
	title        = {Dynamic {Brain} {Networks} with {Prescribed} {Functional} {Connectivity}},
	author       = {Casti, U. and Baggio, G. and Benozzo, D. and Zampieri, S. and Bertoldo, A. and Chiuso, A.},
	year         = 2023,
	month        = dec,
	booktitle    = {2023 62nd {IEEE} {Conference} on {Decision} and {Control} ({CDC})},
	pages        = {709--714},
	doi          = {10.1109/CDC49753.2023.10383872},
	note         = {ISSN: 2576-2370}
}

@article{YU_2007,
	title        = {Parameter identification of dynamical systems from time series},
	author       = {Yu, Wenwu and Chen, Guanrong and Cao, Jinde and Lü, Jinhu and Parlitz, Ulrich},
	year         = 2007,
	month        = jun,
	journal      = {Physical Review E},
	volume       = 75,
	number       = 6,
	pages        = {067201},
	doi          = {10.1103/PhysRevE.75.067201},
	issn         = {1539-3755, 1550-2376},
	copyright    = {http://link.aps.org/licenses/aps-default-license},
	language     = {en}
}

@article{HENDRICKX_2019,
	title        = {Identifiability of {Dynamical} {Networks} {With} {Partial} {Node} {Measurements}},
	author       = {Hendrickx, Julien M. and Gevers, Michel and Bazanella, Alexandre S.},
	year         = 2019,
	month        = jun,
	journal      = IEEE_J_AC,
	volume       = 64,
	number       = 6,
	pages        = {2240--2253},
	doi          = {10.1109/TAC.2018.2867336},
	issn         = {1558-2523}
}

@article{WEERTS_2018,
	title        = {Identifiability of linear dynamic networks},
	author       = {Weerts, Harm H.M. and Van Den Hof, Paul M.J. and Dankers, Arne G.},
	year         = 2018,
	month        = mar,
	journal      = {Automatica},
	volume       = 89,
	pages        = {247--258},
	doi          = {10.1016/j.automatica.2017.12.013},
	issn         = {00051098},
	language     = {en}
}

@inproceedings{BAZANELLA_network_2019,
	title        = {Network identification with partial excitation and measurement},
	author       = {Bazanella, Alexandre S. and Gevers, Michel and Hendrickx, Julien M.},
	year         = 2019,
	month        = dec,
	booktitle    = {2019 {IEEE} 58th {Conference} on {Decision} and {Control} ({CDC})},
	pages        = {5500--5506},
	doi          = {10.1109/CDC40024.2019.9029909},
	note         = {ISSN: 2576-2370}
}

@inproceedings{LOFBERG_2004,
    address = {Taipei, Taiwan},
    author = {L{\"{o}}fberg, J.},
    booktitle = {In Proceedings of the CACSD Conference},
    title = {YALMIP : A Toolbox for Modeling and Optimization in MATLAB},
    year = {2004}
}

@article{GREWAL_1976,
	title        = {Identifiability of linear and nonlinear dynamical systems},
	author       = {Grewal, M. and Glover, K.},
	year         = 1976,
	journal      = IEEE_J_AC,
	volume       = 21,
	number       = 6,
	pages        = {833--837},
	doi          = {10.1109/TAC.1976.1101375}
}

@article{DISTEFANO_1980,
	title        = {On parameter and structural identifiability: Nonunique observability/reconstructibility for identifiable systems, other ambiguities, and new definitions},
	author       = {DiStefano, J. and Cobelli, C.},
	year         = 1980,
	journal      = IEEE_J_AC,
	volume       = 25,
	number       = 4,
	pages        = {830--833},
	doi          = {10.1109/TAC.1980.1102439}
}

@misc{GILL_2025,
      title={Identifying Network Structure of Nonlinear Dynamical Systems: Contraction and Kuramoto Oscillators}, 
      author={Jaidev Gill and Jing Shuang Li},
      year={2025},
      eprint={2509.13505},
      archivePrefix={arXiv},
      primaryClass={eess.SY},
      url={https://arxiv.org/abs/2509.13505}, 
      note = {To appear in 2026 American Control Conference (ACC).}
}

@article{MARDER_2011,
	title        = {Variability, compensation, and modulation in neurons and circuits},
	author       = {Eve Marder},
	year         = 2011,
	journal      = {Proceedings of the National Academy of Sciences},
	volume       = 108,
	number       = {supplement\_3},
	pages        = {15542--15548},
	doi          = {10.1073/pnas.1010674108},
	eprint       = {https://www.pnas.org/doi/pdf/10.1073/pnas.1010674108}
}

@ARTICLE{BASSETT_2017,
  title     = "Small-world brain networks revisited",
  author    = "Bassett, Danielle S and Bullmore, Edward T",
  journal   = "Neuroscientist",
  publisher = "SAGE Publications",
  volume    =  23,
  number    =  5,
  pages     = "499--516",
  month     =  oct,
  year      =  2017,
  copyright = "http://www.sagepub.com/licence-information-for-chorus",
  language  = "en"
}

\addtolength{\textheight}{-12cm}   




\end{document}